\documentclass[11pt,letterpaper,reqno]{amsart}

\usepackage{mathtools}
\usepackage{amsthm,amssymb}
\usepackage{amsmath}
\usepackage[usenames,dvipsnames]{xcolor}
\usepackage{dsfont}
\usepackage{bm}
\usepackage{comment}

\usepackage[shortlabels]{enumitem}
\usepackage{hyperref}
\hypersetup{colorlinks = true, linkcolor = blue, citecolor = blue, urlcolor = blue}
\usepackage{cite}
\allowdisplaybreaks

\usepackage[margin=1in]{geometry}

\usepackage{todonotes}

\usepackage
[
    ruled,         
    vlined,        
    linesnumbered, 
    procnumbered
]
{algorithm2e} 

\usepackage{cleveref}

\newtheorem{theorem}{Theorem}
\newtheorem*{theorem*}{Theorem}

\newtheorem{remark}{Remark}

\newtheorem{lemma}[theorem]{Lemma}

\newtheorem{proposition}[theorem]{Proposition}
\newtheorem{definition}[theorem]{Definition}

\newtheorem*{claim*}{Claim}

\newtheorem{example}[theorem]{Example}
\newtheorem{assumption}[theorem]{Assumption}

\newcommand{\lam}{\lambda}
\newcommand{\Lam}{\Lambda}

\newcommand{\diff}{\mathrm{d}}
\newcommand{\E}{\mathbb{E}}
\newcommand{\R}{\mathbb{R}}

\newcommand{\N}{\mathbb{N}}

\newcommand{\C}{\mathbb{C}}
\newcommand{\D}{\mathbb{D}}
\newcommand{\ind}[1]{\mathds{1}_{#1}}
\newcommand{\absolute}[1]{\left\lvert #1 \right\rvert}

\newcommand{\V}[1]{\mathbf{#1}}
\newcommand{\vol}{\nu}
\newcommand{\borel}{\mathcal{B}}
\newcommand{\X}{\mathbb{X}}
\newcommand{\ball}{\mathbb{B}}
\newcommand{\diam}{\mathsf{diam}}

\DeclareRobustCommand{\stirling}{\genfrac\{\}{0pt}{}}

\title{Simple Analyticity Criteria for Repulsive Multi-body
  Potentials}

\author{Tyler Helmuth}
\address{Durham University, Department of Mathematical Sciences}
\email{tyler.helmuth@durham.ac.uk}

\author{Marcus Pappik}
\address{University of Potsdam, Hasso Plattner Institute}
\email{marcus.pappik@hpi.de}

\author{Will Perkins}
\address{Georgia Institute of Technology, School of Computer Science}
\email{math@willperkins.org}

\begin{document}

\begin{abstract}
  We prove a simple, explicit lower bound on the radius of a zero-free
  disk for Gibbs point processes defined by finite-range, repulsive
  multi-body interactions. Our lower bound improves on those
  previously known, and we demonstrate that it is essentially sharp in
  the generality with which our arguments apply. The key ingredient is
  a multi-body generalization of integral identities for point
  densities of Gibbs point processes in the spirit of earlier work of
  Michelen and Perkins.
\end{abstract}

\maketitle

\section{Introduction}
\label{sec:intro}

Many of the simplest models of statistical mechanics of interacting
particles arise by imposing constraints (or other interactions) on
Poisson point processes.  A well-known example is the \emph{two-body
  hard sphere gas}. In this case, the Poisson process $X$ has points
in $\R^{d}$, and the interaction is a pure constraint: it restricts
$X$ to configurations such that spheres of radius $r>0$ about each
point in $X$ are disjoint. Equivalently, no point in space is within
distance $r$ of two distinct elements of $X$. Despite the simplicity
of this model, relatively little is rigorously understood; see, e.g.,
the introductions
of~\cite{helmuth2022correlation,michelen2023analyticity}
or~\cite[Section~4]{hukushima2025damage}.

This paper is motivated by the following question. Modify the two-body
hard sphere gas to obtain a three-body hard sphere gas. That is,
subject $X\subset \R^{d}$ to the constraint that no point in space is
within distance $r$ of \emph{three} distinct elements of $X$. For a
fixed \emph{a priori} intensity of the underlying Poisson process, it
seems clear that if the two-body hard sphere model is in a dilute
(gaseous) phase, then so must be the three-body hard sphere
model. Unfortunately, known criteria on the intensity that
provably ensure a model is in a dilute phase do not behave in the
expected manner: they yield a larger dilute phase region for the
two-body hard sphere gas than the three-body hard sphere gas. We
discuss this further in \Cref{sec:prior-results} below. In the preceding
sentences the phrase `dilute' should be interpreted as a synonym for
analyticity of the pressure of the models as a function of the
activity of particles (intensity of the point process).

The central issue seems to be that while analyticity criteria for
multi-body continuum particle systems have been explored (see, e.g.,
~\cite{Moraal,procacci2000gas,rebenko2005polymer,rebenko1997convergence,Greenberg1971}),
known results have been primarily concerned with existential
statements. In particular, existing treatments can be analytically
inconvenient and/or inefficient -- they guarantee the existence of
activities small enough that analyticity is true, but without easy
control of precisely how small. Motivated by this, the present paper
develops a simple criteria that readily verifies the intuitively clear
statement above. That is, we give explicit domains of analyticity for
the pressure of multi-body \emph{repulsive} finite-range particle
systems.

Despite the somewhat frivolous question that initiated this work, we
think our results reveal some important questions and perspectives. We
defer our discussion of these matters, as well as connections to
related literature, until after precisely describing our main results.

\subsection{Models and main result}
\label{sec:model-results}

While our results are more general, for notational ease in this
section let $\X = \R^{n}$ be $n$-dimensional Euclidean space.  A
\emph{potential} $\phi$ is a collection of (measurable) functions
$(\phi_k)_{k \in \N}$, $\phi_{k}\colon \X^{k} \to \R\cup
\{\infty\}$. Given $\V{x}\in \X^{k}$, we write
$\phi(\V{x})=\phi_{k}(\V{x})$.  We will always assume that $\phi$ is
\emph{symmetric}, i.e., for each $k \ge 2$,
$\phi_{k}\circ \pi = \phi_{k}$ for all permutation $\pi$ of the $k$
coordinates $\V{x}\in \X^{k}$. We say $\phi$ is \emph{repulsive} if
$\phi_k\geq 0$ for all $k \in \N$.

Given $\V{x}\in \X^{k}$ and $S\subset [k]\coloneqq \{1,2,\dots, k\}$,
let $\V{x}_{S} = (x_{i})_{i\in S}$.  Any $\phi$ defines a Hamiltonian
by
\begin{equation*}
  H(\V{x}) \coloneqq \sum_{S\neq\emptyset} \phi(\V{x}_{S}),
\end{equation*}
where the sum ranges over non-empty subsets of $[k]$ if
$\V{x}\in \X^{k}$.  For $\lambda\in \C$ and $\Lambda \subseteq \X$
bounded and measurable the associated \emph{partition function} is
\begin{equation}
  \label{eq:PF}
    Z_{\Lambda,\phi}(\lam) \coloneqq \sum_{k \in \N_0} \frac{\lam^k}{k!} \int_{\Lambda^k} e^{-H(\V{x})} \vol^{k}(\diff \V{x}),
\end{equation}
where $\vol^{k}$ denotes the $k$-fold product of Lebesgue measure
$\vol$ on $\R^{n}$ with itself. The parameter $\lambda$ is called the
\emph{activity}. The assumption that $\phi$ is repulsive ensures that
$Z_{\Lambda,\phi}$ exists, and moreover,
$|Z_{\Lambda,\phi}|\leq \exp( |\lambda| \vol(\Lambda))$. In
particular, for $\lambda \ge 0$, one obtains a point process $\mu$ in
$\Lambda$ by defining an expectation
\begin{equation*}
    \E_{\mu}[f] \coloneqq \frac{1}{Z_{\Lambda, \phi}(\lambda)} \sum_{k \ge 0} \frac{\lambda^k}{k!} \int_{\Lambda^k} f(\delta_{x_1} + \dots + \delta_{x_k}) e^{-H(\V{x})} \vol^{k}(\diff \V{x}).
\end{equation*}
to every measurable function $f\colon \mathcal{N} \to \R_{\ge 0}$ on the space of locally-finite counting measures $\mathcal{N}$.

We write $\diam(x_1, \dots, x_k) \coloneqq \max_{i, j} d(x_i, x_j)$ for the diameter of a tuple of points, and we say a potential $\phi$ has range (at most) $R$ if $\phi(\V{x}) = 0$ whenever $\diam(\V{x}) > R$. 
Our main result is the following. 
\begin{theorem}
  \label{thm:main} Suppose $\phi$ is repulsive with finite range $R$,
  and let $B_{R}$ denote that volume of a ball of radius $R$. There is
  a $C>0$ such that if $\absolute{\lambda} < (eB_{R})^{-1}$,
  \begin{equation*}
\absolute{\log(Z_{\Lambda, \phi}(\lambda))} \le C \cdot
  \vol(\Lambda).
\end{equation*}
\end{theorem}

\Cref{thm:main} implies that the limiting \emph{pressure}
\begin{equation*}
  p(\lambda) \coloneqq \lim_{\Lambda\uparrow \X} \frac{1}{\vol(\Lambda)} \log
  Z_{\Lambda,\phi}(\lambda) 
\end{equation*}
is an analytic function on the open disk $|\lambda| < (e
B_{R})^{-1}$ whenever the limit exists (e.g., if $\X=\R^d$, $\phi$ is translation invariant,\footnote{We call $\phi$ translation invariant if $\phi(x_1, \dots, x_k) = \phi(x_1 + y, \dots, x_k + y)$ for all $y \in \R^d$.} and we consider a van Hove sequence of regions \cite{ruelle1969statistical}). 
Correspondingly, there are no phase transitions for the
particle system associated to the potential $\phi$ when
$|\lambda|<(eB_{R})^{-1}$. 

\begin{example}
  \label{ex:kbody}
  The \emph{$k$-body hard sphere gas} has only the single non-zero term
  \begin{equation*}
    \phi_{k}(x_{1},\dots, x_{k}) = \begin{cases}
        \infty &\text{ if } \exists y \in \X: \{x_1, \dots, x_k\} \subseteq \ball_y(r) \\
        0 &\text{ otherwise}
    \end{cases},
  \end{equation*}
  where $\ball_y(r)$ denotes the closed ball with radius $r$ and center $y$.
  Thus, $\phi$ forbids $k$ spheres of radius $r$ from overlapping at a single
  point in space. \Cref{thm:main} gives the \emph{same} disk of analyticity
  $|\lambda|\leq (eB_{2r})^{-1}$, irrespective of the value of $k$. 
\end{example}

\subsection{Discussion and Comparison with Existing Results}
\label{sec:further-discussion}

\subsubsection{Prior results}
\label{sec:prior-results}

Prior results have generally imposed either restrictive criteria
(e.g., requiring a hard-core pair
potential~\cite{procacci2000gas,rebenko2005polymer}), imposed analytic
requirements on potentials for technical
purposes~\cite{Moraal,Greenberg1971}, or some combination of
both. The imposition of a hard-core is a natural criteria that removes
one of the central difficulties in the analysis of multi-body
potentials, see \Cref{sec:further-discussion-1} below.

The criteria in \cite[Theorem~2.6]{Greenberg1971} are rather difficult
to apply if one wishes to obtain explicit estimates, even in the case
of repulsive potentials. The criteria in~\cite{Moraal} (see
also~\cite{gielerak1989mayer}), however, are relatively simple for
potentials of the form considered in the present paper, as the
condition of \emph{regularity} used in~\cite{Moraal} is satisfied by
repulsive potentials. This leads to bounds
$|\lambda|\leq (2 e B_{R})^{-1}$, which is only slightly worse than
the bound $(e B_{R})^{-1}$ the Kirkwood-Salsburg method gives for
two-body potentials. Our improvement of removing this factor of two is
relatively modest. However, we think our method more transparently shows the
difficulties in obtaining stronger results.

\subsubsection{Proof techniques}
\label{sec:proof-techniques}
We prove \Cref{thm:main} by establishing a recursive integral identity
for (a modified version of) the point densities of finite-volume Gibbs
point processes with repulsive multi-body potentials. See
\Cref{thm:itegral_identity}.  Provided that this integral identity is
a contraction on a rich enough class of potentials and $|\lambda|$ is
small enough, this yields uniform bounds on the modified point
densities.  \Cref{thm:main} then follows by expressing
$\log(Z_{\Lambda, \phi}(\lambda))$ in terms of modified point
densities.

Our approach is inspired by the work of Michelen and Perkins, who
derived a similar integral identity for repulsive pair
potentials~\cite{michelen2023analyticity,michelen2022strong,michelen2021potential}.
As discussed by Michelen and Perkins, their approach was inspired by
Weitz's tree recursion for discrete spin
systems~\cite{weitz2006counting}. In a similar way our identity can be
viewed as an adaptation of a tree recursion for spin systems on
hypergraphs~\cite{liu2014fptas,bencs2025optimal} to the continuum
setting of Gibbs point processes.

While our integral identity (\Cref{thm:itegral_identity}) is derived
using ideas very similar to those applied in the case of two-body
potentials, it should be noted that the derivation
in~\cite{michelen2023analyticity} uses the classical point density
functions associated to a Gibbs point process.  In contrast, our
identity concerns a modified version of these densities, which allows
for a finer control over the terms that appear at deeper levels of the
recursion. In fact, the modified point densities used in our identity
can be seen as lying midway between the classical point densities and
a type of thermodynamic ratio that was studied for pair potentials in
a largely unnoticed paper by Meeron~\cite{meeron1970bounds}.  We note,
however, that although the recursive identity
in~\cite{meeron1970bounds} has a similar flavor to that
in~\cite{michelen2023analyticity}, it remains unclear if Meeron's
analysis has a useful extension to the multi-body setting.

While for simplicity we have used our method only to obtain disks of
analyticity, our method is not intrinsically restricted to such
domains. One could attempt to find complex domains that contain larger
segments of the positive real axis on which our integral identity is a
contraction. Such a program has been carried out for two-body
potentials~\cite{michelen2023analyticity}, as well as for the
hypergraph tree recursion for discrete spin systems that our
identities mimic~\cite{bencs2025optimal}.

\subsubsection{Further discussion}
\label{sec:further-discussion-1}

\Cref{thm:main} concerns general multi-body potentials, and the domain
of analyticity determined does not depend on either the minimum nor
maximum (if it exists) of the arities $k$ such that $\phi_{k}\neq
0$. One might wonder if a hypothesis that $\phi_{k}=0$ for all
$k\leq K$ can lead to an improvement (in the dependence on range) of bounds.  Somewhat surprisingly,
improved bounds are not possible in generality: in \Cref{sec:sharp} we
give an example that shows the disk of analyticity in \Cref{thm:main}
is (essentially) sharp. To go beyond our results (or those of~\cite{Moraal})
requires exploiting geometric properties of the model being
considered.

The impossibility of improved bounds is inspired (and established) by
analogous results in the discrete setting of hypergraphs. The discrete
analog of the hard sphere gas is the hard-core model (or independent
set model). We recall the relevant definitions in \Cref{sec:sharp};
they are not needed for the present discussion. On a usual graph
(edges all of size $2$) the optimal zero-free disk for the hard-core
partition function over graphs of maximum degree $\Delta$ is known and
has connections to the Lov\'{a}sz Local Lemma from
combinatorics~\cite{shearer1985problem,scott2005repulsive}.  For
$k$-uniform hypergraphs of maximum degree $\Delta$ (each edge of size
$k$), $k \ge 3$, one might expect an improved bound; it has been
proved recently that the zero-free disk is at least as large as that
for graphs~\cite{galvin2024zeroes,bencs2025optimal,neumann2025multi},
but the disk need not be much larger~\cite{zhang2025hypergraph}. This
is in spite of the fact that a dilute phase persists for a much larger
range of \emph{positive} activities
$\lambda$~\cite{hermon2019rapid,bezakova2019approximation,liu2025phase}.
The analogy between the hard-core model on hypergraphs and particles
interacting via a $k$-body interaction is not entirely accurate,
however. Since only one particle is allowed per vertex in a discrete
system, there is an implied short-range $2$-body hard-core interaction
which need not be present in the continuum system.

Relatedly, and recalling our motivating discussion in \Cref{sec:intro}
and Example~\ref{ex:kbody}, it is natural to wonder if the domain of
analyticity for $3$-body hard spheres is strictly greater than for
$2$-body hard spheres. While improvements by detailed calculation may
be feasible, a softer argument that applies to $k$-body hard spheres
would be preferable. It would be even more interesting to have a
non-perturbative argument that establishes that for $k\geq 3$, the
first point of non-analyticity on the positive real activity axis for
$k$-body hard spheres cannot occur before the corresponding point for
$(k-1)$-body hard spheres. That is, $k$-body hard spheres do not
undergo a phase transition before $(k-1)$-body hard
spheres.

One might also wonder if the strength of the potential can be
exploited to obtain improved bounds. To explain this by example,
consider the \emph{$k$-body $\alpha$-soft sphere model}, given by the
pure $k$-body interaction which has the single non-zero potential
\begin{equation*}
    \phi_{k}(x_{1},\dots, x_{k}) = \begin{cases}
        \alpha &\text{ if } \exists y \in \X: \{x_1, \dots, x_k\} \subseteq \ball_y(r) \\
        0 &\text{ otherwise}
    \end{cases},
  \end{equation*}
for some $\alpha \in (0,\infty)$. Intuitively, the low-density phase
of the $k$-body $\alpha$-soft sphere model should occur in a strictly
larger domain than for the $k$-body hard sphere model, and one might
expect the same is true for domains and/or disks of analyticity. Our
results do not establish this, however. The fundamental difficulty is
that there are configurations of particles in which there are many
particles at essentially the same location, and the potential created
by these particles can approximate the $k$-body hard sphere potential
arbitrarily well. It seems difficult to exploit the fact that such
configurations are entropically unlikely (or that there are associated
helpful $1$-body potentials created as well) with our method. 

It is perhaps worth remarking that in the special case of a $2$-body
interaction, it is well-known that a soft interaction \emph{does} lead
to an improved disk of analyticity via the temperedness constant. See,
e.g.,~\cite[Section~1.2]{michelen2023analyticity}.  In this special
case the improvement is straightforward since the potentials created
by particles are $1$-body potentials, which only have the effect of
decreasing the activity parameter. A method for analyzing $k$-body
potentials which clearly exhibited similar improvements would be
interesting.

\section{General setting and definitions}

Let $(\X, d)$ be a complete, separable metric space, and let $\borel$
be the Borel algebra generated by $d$.  Equip $(\X, \borel)$ with a
reference measure $\vol$, which we refer to as \emph{volume}. We will
always assume $\vol$ is locally finite, i.e., assigns finite volume to
bounded Borel sets.

As indicated earlier, the basic example to have in mind is
$\X=\R^{n}$, $d$ the Euclidean metric. For more general $\X$ we will
need to impose one ``smoothness'' assumption that Euclidean space
possesses. For every $x \in \X$, let $\vol_x$ denote the push-forward of $\vol$
under the map $y \mapsto d(x, y)$, i.e., for $U \subseteq \R$
measurable, 
$\vol_x(U) = \vol(\{y \in \X \mid d(x, y) \in U\})$.
\begin{assumption}
  \label{as:continuity}
  We assume $(X,d,\vol)$ has the property that for all $x\in \X$, the
  measure $\vol_x$ is absolutely continuous with respect to Lebesgue
  measure on $\R$.
\end{assumption}

In this generality we define potentials and partition functions
exactly as in \Cref{sec:model-results}; the only change is the
interpretation of the metric and volume
measure. \Cref{lem:as:continuity} below will describe how
\Cref{as:continuity} will be used.

\subsection{Pinnings and Modified Point Densities} \label{sec:densities}
Given a potential $\phi$ and a point $x \in \X$, we define the
\emph{potential with pinning $x$}, denote $\phi(\cdot \mid x)$, by
\begin{equation*}
    \phi(\V{y} \mid x) \coloneqq \phi(\V{y}) + \phi(\V{y}, x), \quad
    \V{y} \in \X^m, m \in \N. 
\end{equation*}
In this equation we have slightly abused notation;
$\phi(\V{y},x) = \phi( (\V{y},x))$ where $(\V{y},x)$ is interpreted as
an element of $\X^{m+1}$. We extend this modification to tuples of
points $\V{x}\in \X^{k}$ by
\begin{equation*}
  \phi(\V{y} \mid \V{x}) 
  \coloneqq 
  \sum_{S \subseteq [k]} \phi(\V{y}, \V{x}_{S}), 
\end{equation*}
where we have again abused notation slightly.  Denote by $H(\cdot
\mid \V{x})$ the Hamiltonian associated with $\phi(\cdot \mid
\V{x})$ and write $Z_{\Lambda, \phi}(\lambda \mid
\V{x})$ for the partition function on
$\Lambda$ with activity $\lambda\in
\C$ and modified potential $\phi(\cdot \mid \V{x})$.

Provided $Z_{\Lambda,\phi}(\lambda)\neq 0$, define \emph{modified
  (complex) $k$-point densities} by
\begin{equation}
  \label{eq:mkptdensity}
    \kappa_{\Lambda, \phi, \lambda}(\V{x}) \coloneqq \lambda^k \cdot \frac{Z_{\Lambda, \phi}(\lambda \mid \V{x})}{Z_{\Lambda, \phi}(\lambda)} ,
\end{equation}
where we allow ourselves to omit $\Lambda, \phi$ or $\lambda$ from the
notation if they are clear from the context.  Using a telescopic
product, modified $k$-point densities for $\V{x} = (x_1, \cdots, x_k)$
can be written as products of modified one-point densities:
\begin{align} \label{eq:telescoping}
    \kappa(\V{x}) 
    = \prod_{j = 1}^{k} \lambda \cdot \frac{Z_{\Lambda, \phi}(\lambda \mid x_1, \dots, x_{j})}{Z_{\Lambda, \phi}(\lambda \mid x_1, \dots, x_{j-1})} 
    =\prod_{j = 1}^{k} \kappa(x_j \mid x_1, \dots, x_{j-1}),
\end{align}
where $\kappa(x_j \mid x_1, \dots,
x_{j-1})$ are the modified one-point densities at
$x_j$ under the potentials $\phi(\cdot \mid x_1, \dots,
x_{j-1})$. The $j=1$ term is to be understood as $\kappa(x_1)$.

\begin{remark}
    For $\lambda\geq 0$, the modified densities are related to the
    classical $k$-point densities (sometimes also called correlation functions) via $\rho(\V{x}) = e^{-H(\V{x})}
    \cdot \kappa(\V{x})$, see~\cite[Chapter 4.1.1]{ruelle1969statistical}.
\end{remark}

\section{An integral identity for modified point densities}
The goal of this section is to establish \Cref{thm:itegral_identity},
an integral identity for modified point densities. The correct
formulation of the identity requires several preliminary definitions.

\subsubsection*{Ordering and Partial Pinning}
Fix an arbitrary $z\in \X$. Given a tuple of points $\V{x} = (x_1, \dots, x_k) \in \X^{k}$, define 
\begin{equation}
  \label{eq:Dchoice}
    D(\V{x}) \coloneqq \sum_{j \in [k]} d(z, x_j).
\end{equation}
The specific choice of $D$ in~\eqref{eq:Dchoice} is not crucial; what
is important are the following properties of the restriction $D_{k}$
of $D$ to $\X^{k}$.  The proof of the lemma will be deferred
to~\Cref{apx:ordering}.
\begin{lemma}
  \label{lemma:ordering}
  If $(\X, d, \vol)$ satisfies \Cref{as:continuity}, then for every $k\in\N$:
  \begin{enumerate}[(1), noitemsep]
  \item $D_k$ is symmetric and
    measurable. \label{lemma:ordering:measurability}
  \item 
    The push-forward of $\vol^k$ under $D_k$ is absolutely continuous
    with respect to the Lebesgue measure on $\R_{\ge 0}$.  
  \label{lemma:ordering:continuity}
    \item\label{lemma:ordering:maximum} For every $k \in \N$ and $\vol^k$-almost all
    $\V{x} \in \X^k$, it holds that $D$ induces a strict total order on subsets of $[k]$ by setting $T < S$ iff $D(\V{x}_T) < D(\V{x}_S)$.
  \end{enumerate}
\end{lemma}

The significance of part \ref{lemma:ordering:continuity} of the preceding lemma is the following.
\begin{lemma}
  \label{lem:as:continuity}
  Under \Cref{as:continuity}, for each $k\geq 1$, there is a family of measures
  $\{\vol_{t}^{k}\}_{t \in \R_{\ge 0}}$ such that
  \begin{enumerate}[noitemsep]
  \item $\vol^{k}_{t}$ assigns measure zero to the set $\{\V{w}\in
    \X^{k} \mid D(\V{w})\neq t\}$, and 
    \item there is a density $g_{k}$ on $\R_{\ge 0}$ such that
    $\int_{\X^{k}} h(\V{w}) \vol^k(\diff \V{w}) = \int_{\R_{\ge 0}} \int_{\X^{k}} h(\V{w})
    \vol^k_t(\diff \V{w}) g_k(t) \diff t$ for all $h\colon \X^{k} \to \C$ such
    that the integral exists.
  \end{enumerate}
\end{lemma}
\begin{proof}
  This is the disintegration theorem for the Lebesgue measure combined 
  with part \ref{lemma:ordering:continuity} of \Cref{lemma:ordering}. 
\end{proof}

For tuples $\V{x} \in \X^{k}, \V{y} \in \X^{m}$ and
$t \in [0, \infty]$ we write $\V{x} \prec t$ if $D(\V{x}) < t$ and
$\V{x} \prec \V{y}$ if $D(\V{x}) < D(\V{y})$.  Moreover, given
$y \in \X, \V{x} \in \X^k, \V{w} \in \X^m$, we use these relations to
define the modified potentials
\begin{align*}
    \phi(\V{x} \mid y_{\prec t})
        \coloneqq \begin{cases}
            \phi(\V{x} \mid y) &\text{ if $\V{x} \prec t$} \\
            \phi(\V{x}) &\text{ otherwise}
        \end{cases}, \qquad 
    \phi(\V{x} \mid y_{\prec \V{w}}) &\coloneqq \begin{cases}
        \phi(\V{x} \mid y) &\text{ if } \V{x} \prec \V{w} \\
        \phi(\V{x}) &\text{ otherwise.}
    \end{cases} 
\end{align*} 
Note that $\phi(\cdot \mid y) = \phi(\cdot \mid y_{\prec \infty})$ and
$\phi(\cdot \mid y_{\prec \V{w}}) = \phi(\cdot \mid y_{\prec
  D(\V{w})})$. We think of these modified potentials as being
\emph{partial pinnings}: the pinning at $\V{y}$ is imposed on configurations
that precede $\V{w}$ (according to $D$). 

As in \Cref{sec:densities} for pinned potentials, we extend this
notation to Hamiltonians, partition functions and modified point
densities.  For example, we write $H(\cdot \mid y_{\prec \V{w}})$,
$Z_{\Lambda, \phi}(\cdot \mid y_{\prec \V{w}})$ and
$\kappa(\cdot \mid y_{\prec \V{w}})$ for the Hamiltonian, partition
function and modified point density associated with the potential
$\phi(\V{x} \mid y_{\prec \V{w}})$.

Occasionally, it will be necessary to iterate pinnings/partial
pinnings. When this occurs, we separate the modifications we make with
a semicolon, and the modifications are made from left to right.  For
example,
$\phi(\cdot \mid y_{\prec \V{w}}; \V{z}) = \psi(\cdot \mid \V{z})$ is
the pinning of $\psi = \phi(\cdot \mid y_{\prec \V{w}})$ at $\V{z}$.
As before, we extend this notation to Hamiltonians, partition
functions, and point densities.

\subsubsection*{Sets of Potentials}
Recall that modified $k$-point densities are only well-defined when
the associated partition function is non-zero. To be able to discuss
$k$-point densities for a class of potentials, we define
the notion of being $\delta$-zero-free.

\begin{definition}
  \label{def:total_zero_free} Let
  $\delta > 0$. A collection of potentials $\Phi$ is
  \emph{$\delta$-zero-free} for a bounded and measurable
  $\Lambda \subset \X$ and an activity $\lambda \in \C$ if
  $|Z_{\Lambda, \phi}(\lambda)| \ge \delta$ for all $\phi\in\Phi$.
\end{definition}
  
Our integral identity will relate $k$-point densities to $k$-point
densities associated with modified potentials.  Accordingly, we
introduce the following definition.
\begin{definition}
  \label{def:closed_under_modification}
  A collection of potentials $\Phi$ is \emph{closed under
    modification} if
    \begin{enumerate}[(1),noitemsep]
    \item for every $\phi \in \Phi$, $y \in \X$ and $t \in [0, \infty]$ it holds that
      $\phi(\cdot \mid y_{\prec t}) \in \Phi$, and
    \item for all $\phi \in \Phi$ and measurable $\Delta \subseteq
      \X$, the potentials $\psi$ with   $\psi_k = \phi_k$ for all $k \ge 2$ and
        \begin{equation*}
            \psi_1(x) = \begin{cases}
                \infty &\text{ if $x \in \Delta$}\\
                \phi_1(x) &\text{ otherwise }
            \end{cases}
        \end{equation*}
        are contained in $\Phi$. 
    \end{enumerate}
\end{definition}

\begin{remark}
    \label{rem:closed}
    Property~(2) is a consistency condition that allows for the
    imposition of boundary conditions.  Property~(1) ensures that the
    $\Phi$ is closed under (partial) pinnings of a single point, using
    that, for every $y \in \X$,
    $\phi(\cdot \mid y) = \phi(\cdot \mid y_{\prec \infty})$.  This
    further extends to pinnings of tuples $\V{y} \in \X^k$ via the
    recursion
    $\phi(\V{x} \mid \V{y}) = \psi(\V{x} \mid y_k) \text{ for }
    \psi(\V{x}) \coloneq \phi(\V{x} \mid \V{y}_{[k-1]})$.
\end{remark}

Sets of potentials that are closed under modification satisfy a useful
continuity property, summarized in the next lemma. 
\begin{lemma}\label{lemma:continuity}
    Let $\Phi$ be a collection of repulsive potentials, and let $U \subset \C$ be bounded.
    For every bounded, measurable $\Lambda \subset \X$, the family of functions
    \[
        \{U \to \C, \lambda \mapsto Z_{\Lambda, \phi}(\lambda) \mid \phi \in \Phi\}
    \]
    is uniformly equicontinuous.
    Moreover, if $\Phi$ is closed under modification and $\delta$-zero-free on $\Lambda$ at some activity $\lambda_0 \in \C$ for some $\delta > 0$, then there exists some complex neighborhood $W$ of $\lambda_0$, such that the family of functions
    \[
        \{W \to \C, \lambda \mapsto \kappa_{\Lambda, \phi, \lambda}(x) \mid \phi \in \Phi, x \in \Lambda\}
    \]
    is uniformly equicontinuous.
\end{lemma}
\begin{proof}
  Set $M = \sup\{\absolute{\lambda} \mid \lambda \in U\}$. Note that
  for $\lambda_1, \lambda_2 \in U$,
  $|\lambda_{1}^{k}-\lambda_{2}^{k}|\leq
  kM^{k-1}|\lambda_{1}-\lambda_{2}|$.  Hence for all $\phi \in \Phi$
  and $\lambda_1, \lambda_2 \in U$ it holds that
    \begin{equation*}
        \absolute{Z_{\Lambda, \phi}(\lambda_1) - Z_{\Lambda, \phi}(\lambda_2)} 
            \le \sum_{k \in \N} \frac{\absolute{\lambda_1^k - \lambda_2^k}}{k!} \cdot \vol(\Lambda)^k 
            \leq \absolute{\lambda_1 - \lambda_2} \cdot \vol(\Lambda) \cdot e^{M \vol(\Lambda)}.
    \end{equation*}
    This proves the first statement.  For the second part, let $U$ be
    any bounded neighborhood of $\lambda_0$.  Using the fact that
    $\Phi$ is $\delta$-zero-free at $\lambda_0$ and the first part of
    the lemma, we can find some neighborhood $W \subseteq U$ such that
    \begin{equation*}
        \{W \to \C, \lambda \mapsto Z_{\Lambda, \phi}(\lambda) \mid \phi \in \Phi\}
    \end{equation*}
    is uniformly equicontinuous and uniformly bounded from below by $\delta/2$.
    By the definition~\eqref{eq:mkptdensity} of the modified point density and the
    fact that $\Phi$ is closed under modification, $\kappa_{\Lambda,
      \phi, \lambda}(x)$ is ($\lambda$ times)  a ratio of functions in
    this last set. This shows that 
    \begin{equation*}
        \{W \to \C, \lambda \mapsto \kappa_{\Lambda, \phi, \lambda}(x) \mid \phi \in \Phi, x \in \Lambda\}
    \end{equation*}
    is uniformly equicontinuous as well.
\end{proof}

\subsubsection*{Integral Identity} The main result of this section is
the following theorem.
\begin{theorem}
  \label{thm:itegral_identity}
  Suppose $(\X, d, \vol)$ satisfies \Cref{as:continuity}.  Let
  $\Lambda \subseteq \X$ be bounded and measurable, $\lambda \in \C$
  and let $\Phi$ be a collection of symmetric and repulsive potentials
  that are closed under modification.  If $\Phi$ is $\delta$-zero-free
  for these parameters, then for all $\phi \in \Phi$ the associated
  modified point densities satisfy
  \begin{equation}
    \label{eq:itegral_identity}
    \kappa(y) = \lambda \cdot \exp \left(- \sum_{k \in \N} \frac{1}{k!} \int_{\Lambda^k} \left(1 - e^{-\phi(y, \V{w})}\right) \cdot e^{-H(\V{w} \mid y_{\prec \V{w}})} \cdot \kappa(\V{w} \mid y_{\prec \V{w}}) \vol^k(\diff \V{w})\right) .
  \end{equation}
\end{theorem}

To prove \Cref{thm:itegral_identity} we need three
lemmas.  The first lemma is a fundamental theorem of calculus. We were unable
to find a reference for this (surely well-known) result, and hence
provide a proof in \Cref{sec:lemftc}.
\begin{lemma}
  \label{lemma:ftc} Let $t > 0$, and let
  $f\colon [0, t] \to \C \setminus \{0\}$ be absolutely continuous.  It holds
  that
  \begin{equation*}
    \frac{f(t)}{f(0)} = \exp\left(\int_{0}^{t} \frac{f'(s)}{f(s)} \diff s \right),
  \end{equation*}
  where the integral should be understood as a Lebesgue integral and $f'$ is almost everywhere a derivative of $f$.
\end{lemma}

The second lemma is elementary.
\begin{lemma}
  \label{lem:orderproduct}
  Suppose $A$ is a totally ordered finite set with total order $<$,
  and let $\{x_{a}\}_{a\in A}$ be indeterminates. Then
  \begin{equation}
    \label{eq:orderproduct}
    \prod_{a\in A}(1+x_{a}) = 1+\sum_{b\in A}x_{b}\prod_{a<b}(1+x_{a}).
  \end{equation}
\end{lemma}
\begin{proof}
  Apply the binomial theorem. Collect terms according to the largest
  $b$ such that $x_{b}$ occurs. 
\end{proof}

\begin{lemma} \label{lemma:partition_function_identity}
    Consider the setting of \Cref{thm:itegral_identity}.
    For all $t \in [0, \infty]$ it holds that
    \begin{equation*}
        Z_{\Lambda, \phi}(\lambda \mid y_{\prec t}) = Z_{\Lambda,
          \phi}(\lambda) - \sum_{k \in \N} \frac{\lambda^k}{k!}
        \int_{\Lambda^k} \ind{\V{w} \prec t} \cdot \left( 1 -
          e^{-\phi(y, \V{w})} \right) \cdot e^{-H(\V{w} \mid y_{\prec
            \V{w}})} Z_{\Lambda, \phi}(\lambda \mid y_{\prec \V{w}};
        \V{w}) \vol^k(\diff \V{w}). 
    \end{equation*}
\end{lemma}

\begin{proof}
    Recall that
    \begin{equation*}
        Z_{\Lambda, \phi}(\lambda \mid y_{\prec t}) = 1 + \sum_{m \in \N} \frac{\lambda^m}{m!} \int_{\Lambda^m} e^{-H(\V{x} \mid y_{\prec t})} \vol^m(\diff \V{x}) .
      \end{equation*}
    Moreover, note that for all $m \in \N$ and $\V{x} \in \Lambda^m$ it holds that
    \begin{equation*}
      e^{-H(\V{x} \mid y_{\prec t})} 
        = e^{-H(\V{x})} \prod_{S \subseteq [m]: S \neq \emptyset} \left( 1 - \ind{\V{x}_S \prec t} \cdot \left(1 - e^{- \phi(y, \V{x}_S)}\right)\right).
      \end{equation*}
      By part~\ref{lemma:ordering:maximum} of  \Cref{lemma:ordering},
      for $\vol^m$-almost all $\V{x}$, the set of $S\subset [m]$ are
      totally ordered according to whether or not $\V{x}_{S}\prec
      \V{x}_{S'}$. Hence by applying \Cref{lem:orderproduct},
      \begin{align*}
        e^{-H(\V{x}\mid y_{\prec t})}
        &= e^{-H(\V{x})} -  e^{-H(\V{x})}\sum_{B\subset [m]: B\neq
          \emptyset}
          \ind{\V{x}_{B}\prec t}(1-e^{-\phi(y,\V{x}_{B})})
          \prod_{S\neq \emptyset: \V{x}_{S}\prec \V{x}_{B}}
          (1-\ind{\V{x}_{S}\prec  t}(1-e^{-\phi(y,\V{x}_{S})})) \\
        &= e^{-H(\V{x})}
          -  e^{-H(\V{x})}\sum_{B\subset [m]: B\neq \emptyset}
          \ind{\V{x}_{B}\prec t}(1-e^{-\phi(y,\V{x}_{B})})
          \prod_{S\neq \emptyset: \V{x}_{S}\prec \V{x}_{B}} e^{-\phi(y,\V{x}_{S})} \\
        &= e^{-H(\V{x})} - \sum_{B\subset [m]:B\neq\emptyset}
          \ind{\V{x}\prec  t}(1-e^{-\phi(y,\V{x}_{B})})e^{-H(\V{x}|y_{\prec \V{x}_{B}})}
     \end{align*} 
     by the definition of $H(\V{x}| y_{\prec \V{x}_{B}})$. 
    Substituting this back into the definition of $Z_{\Lambda, \phi}(\lambda \mid y_{\prec t})$ yields
    \begin{equation*}
        Z_{\Lambda, \phi}(\lambda \mid y_{\prec t}) = Z_{\Lambda, \phi}(\lambda) - \sum_{m \in \N} \frac{\lambda^m}{m!} \sum_{\substack{B \subseteq [m]:\\ B \neq \emptyset}} \int_{\Lambda^m} \ind{\V{x}_B \prec t} \cdot \left(1 - e^{-\phi(y, \V{x}_B)}\right) \cdot e^{-H\left(\V{x} \mid y_{\prec \V{x}_{B}}\right)} \vol^{m}(\diff \V{x}) .
      \end{equation*}
      The value of the integral only depends on the cardinality of
      $B$, and hence the last term above can be rewritten as
      \begin{equation*}
        \sum_{m \in \N} \frac{\lambda^m}{m!} \sum_{k=1}^{m} \binom{m}{k} \int_{\Lambda^k} \ind{\V{w} \prec t} \cdot \left(1 - e^{-\phi(y, \V{w})}\right) \int_{\Lambda^{m-k}} e^{-H(\V{w}, \V{x} \mid y_{ \prec \V{w}})} \vol^{m - k}(\diff \V{x}) \vol^{k}(\diff \V{w}) .
      \end{equation*}
    Changing the order of summation, setting $\ell = m - k$, and
    exchanging summation and integration rewrites this quantity as
    \begin{align*}
            \sum_{k \in \N} \frac{\lambda^k}{k!} \int_{\Lambda^k} \ind{\V{w} \prec t} \cdot \left(1 - e^{-\phi(y, \V{w})}\right) \sum_{\ell \ge 0} \frac{\lambda^{\ell}}{\ell!} \int_{\Lambda^{\ell}} e^{-H(\V{w}, \V{x} \mid y_{\prec \V{w}})} \vol^{\ell}(\diff \V{x}) \vol^{k}(\diff \V{w}),
    \end{align*}
    and thus
    \begin{align*}
        Z_{\Lambda, \phi}(\lambda \mid y_{\prec t}) = Z_{\Lambda, \phi}(\lambda) - \sum_{k \in \N} \frac{\lambda^k}{k!} \int_{\Lambda^k} \ind{\V{w} \prec t} \cdot \left(1 - e^{-\phi(y, \V{w})}\right) \sum_{\ell \ge 0} \frac{\lambda^{\ell}}{\ell!} \int_{\Lambda^{\ell}} e^{-H(\V{w}, \V{x} \mid y_{\prec \V{w}})} \vol^{\ell}(\diff \V{x}) \vol^{k}(\diff \V{w}).
    \end{align*}
    The exchanges summation and integration above can be justified
    by using that $\phi$ is repulsive.
    Finally, observing that
    \[
        H(\V{w}, \V{x} \mid y_{\prec \V{w}}) = H(\V{w} \mid y_{\prec \V{w}}) + H(\V{x} \mid y_{\prec \V{w}}; \V{w})
    \]
    yields
    \[
        \sum_{\ell \ge 0} \frac{\lambda^{\ell}}{\ell!} \int_{\Lambda^{\ell}} e^{- H(\V{w}, \V{x} \mid y_{\prec \V{w}})} \vol^{\ell}(\diff \V{x}) = e^{- H(\V{w}\mid y_{\prec \V{w}})} \cdot Z_{\Lambda,\phi}(\lambda \mid y_{\prec \V{w}}; \V{w}),
    \]
    which concludes the proof.
  \end{proof}

\begin{proof}[Proof of \Cref{thm:itegral_identity}]
  Our first step is to obtain an expression for the derivative
  $\frac{\partial Z_{\Lambda, \phi}(\lambda \mid y_{\prec
      t})}{\partial t}$ on $\R_{\ge 0}$. We will do this by
  manipulating the formula obtained in
  \Cref{lemma:partition_function_identity}. By applying
  \Cref{lem:as:continuity} for each $k$, 
  \begin{align*}
    &\sum_{k \in \N} \frac{\lambda^k}{k!} \int_{\Lambda^k} \ind{\V{w} \prec t} \cdot \left( 1 - e^{-\phi(y, \V{w})} \right) \cdot e^{-H(\V{w} \mid y_{\prec \V{w}})} Z_{\Lambda, \phi}(\lambda \mid y_{\prec \V{w}}; \V{w}) \vol^k(\diff \V{w})  \\
    &\hspace{3em}= \int_{0}^{t} \sum_{k \in \N} \frac{\lambda^k}{k!} \int_{\Lambda^k} \ind{\V{w} \prec t} \cdot \left( 1 - e^{-\phi(y, \V{w})} \right) \cdot e^{-H(\V{w} \mid y_{\prec \V{w}})} Z_{\Lambda, \phi}(\lambda \mid y_{\prec \V{w}}; \V{w}) \vol_s^k(\diff \V{w})  g_k(s) \diff s .
  \end{align*}
  The exchange of summation and integration is allowed by absolute
  convergence since the potentials are repulsive (see discussion
  after~\eqref{eq:PF}). 
  Substituting this back into \Cref{lemma:partition_function_identity}
  and applying the fundamental theorem of calculus for Lebesgue-integrable
  functions yields that
  $t \mapsto Z_{\Lambda, \phi}(\lambda \mid y_{\prec t})$ is
  absolutely continuous on bounded intervals in $\R_{\ge 0}$ and
  almost-everywhere differentiable, with
  \begin{align}
    \frac{\partial Z_{\Lambda, \phi}(\lambda \mid y_{\prec t})}{\partial t} = - \sum_{k \in \N} \frac{\lambda^k}{k!} \int_{\Lambda^k} \left( 1 - e^{-\phi(y, \V{w})} \right) \cdot e^{-H(\V{w} \mid y_{\prec \V{w}})} \cdot Z_{\Lambda, \phi}(\lambda \mid y_{\prec \V{w}}; \V{w}) \vol_t^k(\diff \V{w}) g_k(t). \label{eq:derivative_partition_function}
  \end{align}

  Observe that $\phi(\cdot \mid y_{\prec 0}) = \phi$ and
  $\phi(\cdot \mid y_{\prec \infty}) = \phi(\cdot \mid y)$.  Moreover,
  it holds that
  $Z_{\Lambda, \phi}(\lambda \mid y_{\prec t}) \to Z_{\Lambda,
    \phi}(\lambda \mid y_{\prec \infty})$ as $t \to \infty$.  Hence,
  by continuity of the exponential function and dominated convergence,
  the conclusion of the lemma follows if, for every $t \in \R_{>0}$,
  \begin{equation*}
    \frac{Z_{\Lambda, \phi}(\lambda \mid y_{\prec t})}{Z_{\Lambda,
        \phi}(\lambda \mid y_{\prec 0})} = \exp \left(- \sum_{k \in
        \N} \frac{1}{k!} \int_{\Lambda^k} \ind{\V{w} \prec t} \left(1
        - e^{-\phi(y, \V{w})}\right) \cdot e^{-H(\V{w} \mid y_{\prec
          \V{w}})} \cdot \kappa(\V{w} \mid y_{\prec \V{w}})
      \vol^k(\diff \V{w})\right).
  \end{equation*}
  Towards this, observe that since
  $s \mapsto Z_{\Lambda, \phi}(\lambda \mid y_{\prec s})$ is non-zero
  and absolutely continuous on bounded intervals, \Cref{lemma:ftc}
  yields
  \begin{equation*}
    \frac{Z_{\Lambda, \phi}(\lambda \mid y_{\prec t})}{Z_{\Lambda,
        \phi}(\lambda \mid y_{\prec 0})}
    = \exp\left( - \int_0^t \frac{1}{Z_{\Lambda, \phi}(\lambda \mid y_{\prec s})} \cdot \frac{\partial Z_{\Lambda, \phi}(\lambda \mid y_{\prec s})}{\partial s} \diff s\right) .
  \end{equation*}
  Moreover, by \Cref{eq:derivative_partition_function} and the fact
  that $\vol_s^k (\diff \V{w}) g_k(s) \diff s = \vol^k(\diff \V{w})$
  and
  $Z_{\Lambda, \phi}(\lambda \mid y_{\prec s}) = Z_{\Lambda,
    \phi}(\lambda \mid y_{\prec \V{w}})$ for $\vol_s^k$-almost all
  $\V{w} \in \Lambda^k$, we get
  \begin{align*}
    & \int_0^t \frac{1}{Z_{\Lambda, \phi}(\lambda \mid y_{\prec s})} \cdot \frac{\partial Z_{\Lambda, \phi}(\lambda \mid y_{\prec s})}{\partial s} \diff s \\
    &\hspace{1em} = 
      \int_0^{t} \sum_{k \in \N} \frac{\lambda^k}{k!} \int_{\Lambda^k} \left( 1 - e^{-\phi(y, \V{w})} \right) \cdot e^{-H(\V{w} \mid y_{\prec \V{w}})} \cdot \frac{Z_{\Lambda, \phi}(\lambda\mid y_{\prec \V{w}}; \V{w})}{Z_{\Lambda, \phi}(\lambda \mid y_{\prec s})} \vol^k_s(\diff \V{w}) g_k(s) \diff s\\
    &\hspace{1em}= \sum_{k \in \N} \frac{1}{k!} \int_{\Lambda^k} \ind{\V{w} \prec t} \left( 1 - e^{-\phi(y, \V{w})} \right) \cdot e^{-H(\V{w} \mid y_{\prec \V{w}})} \cdot \kappa(\V{w} \mid y_{\prec \V{w}}) \vol^k(\diff \V{w}) ,
  \end{align*}
  where the exchange of integration and summation is justified by
  dominated convergence, using the fact that
  $\V{w} \mapsto \kappa(\V{w} \mid y_{\prec \V{w}})$ is bounded given
  that $\Phi$ is closed under modification and $\delta$-zero-free on
  $\Lambda$ at activity $\lambda$.
\end{proof}

\begin{remark}
  \label{rem:connected_parts}
  One approach to proving zero-freeness for partition functions is to
  derive a ``connected parts'' identity that shows $\log Z$ can be
  expressed as a sum over terms indexed by trees (for two-body
  potentials) or generalized tree-like objects~\cite{brydges1986}. We
  view~\eqref{eq:itegral_identity} as a variation on this theme:
  $\kappa(y)$ is a derivative of $\log Z$, and our formula locally
  explores the tree-like structure of the particle at $y$. Our choice
  of $D$ in~\eqref{eq:Dchoice} plays the role of the choice of a
  partition scheme, see, e.g.,~\cite{UeltschiOWR}. By avoiding explicit sums over
  tree-like objects we avoid convergence issues related to their
  enumeration. 
\end{remark}

The following lemma has a proof similar to that of
\Cref{thm:itegral_identity}; the conjunction of the lemmas is what
will allow us to establish zero-free disks for partition functions.
\begin{lemma}
  \label{lemma:partitionfunctionFromDensity}
  Suppose $(\X, d, \vol)$ satisfies \Cref{as:continuity}, and let
  $z\in \X$.\footnote{This $z$ does not need to be the same as the one
    used in the definition~\eqref{eq:Dchoice} of $D$.}  Let
  $\Lambda \subset \X$ be bounded and measurable, $\lambda \in \C$ and
  let $\Phi$ be a collection of repulsive potential that is closed
  under modification.  For every $\phi \in \Phi$ and $x \in \X$,
  define the potential $\phi^{(x)}$ with $\phi^{(x)}_k = \phi_k$ for
  $k \ge 2$ and
    \begin{equation*}
        \phi^{(x)}_1(y) \coloneqq \begin{cases}
            \infty &\text{ if } d(z, y) < d(z, x) \\
            \phi_1(y) &\text{ otherwise}.
        \end{cases}
    \end{equation*}
    Note that $\phi^{(x)} \in \Phi$.  If $\Phi$ is $\delta$-zero-free
    on $\Lambda$ at activity $\lambda \in \C$ for some $\delta > 0$,
    then, for all $\phi \in \Phi$, it holds that
    \begin{equation*}
        \log(Z_{\Lambda, \phi}(\lambda)) = \int_{\Lambda} e^{-\phi^{(x)}(x)} \kappa^{(x)}(x) \vol(\diff x),
    \end{equation*}
    where $\kappa^{(x)}$ is the modified point density associated with $\phi^{(x)}$ on $\Lambda$ at activity $\lambda$.
\end{lemma}

\begin{proof}
    For $t \in [0, \infty]$ define the potential function $\phi^{(t)}$ with $\phi^{(t)}_k = \phi_k$ for all $k \ge 2$ and
    \[
        \phi^{(t)}_1(y) \coloneqq \begin{cases}
            \infty \text{ if } d(z, y) < t \\
            \phi_1(y) \text{ otherwise }
        \end{cases}
    \]
    for the same $z \in \X$ as in the definition of $\phi^{(x)}$.
    Write $H^{(x)}$ and $H^{(t)}$ for the energy functions associated with $\phi^{(x)}$ and $\phi^{(t)}$.
    Moreover, define $\Lambda_x = \{y \in \Lambda \mid d(z, y) < d(z, x)\}$.
    We start by observing that, for every $m \ge 1$, it holds that
    \begin{align*}
        \int_{\Lambda^m} \lambda^m e^{-H^{(t)}(\V{z})} \vol^m(\diff \V{z})
        &= m \int_{\Lambda} \int_{(\Lambda_x)^{m-1}} \lambda^m e^{-H^{(t)}(x, \V{z})} \vol^{m-1}(\diff \V{z}) \vol(\diff x) \\
        &= m \int_{\Lambda} \lambda e^{-\phi^{(t)}(x)} \int_{(\Lambda_x)^{m-1}} \lambda^{m-1} e^{-H^{(t)}(\V{z} \mid x)} \vol^{m-1}(\diff \V{z}) \vol(\diff x)
        \\
        &= m \int_{\Lambda} \lambda e^{-\phi^{(t)}(x)} \int_{\Lambda^{m-1}} \lambda^{m-1} e^{-H^{(x)}(\V{z} \mid x)} \vol^{m-1}(\diff \V{z}) \vol(\diff x),
    \end{align*}
    where the first equality follows from the fact that $H^{(t)}$ is
    invariant under permutation of points. 
    Hence 
    \begin{align*}
        Z_{\Lambda, \phi^{(t)}}(\lambda) - 1 
        = \sum_{m \in \N} \frac{\lambda^m}{m!} \int_{\Lambda^m} e^{-H^{(t)}(\V{z})} \vol^{m}(\diff \V{z}) 
        = \int_{\Lambda} \lambda e^{-\phi^{(t)}(x)} Z_{\Lambda, \phi^{(x)}}(\lambda \mid x) \vol(\diff x).
    \end{align*}
    The last integral can be written as a double integral by
    \Cref{lem:as:continuity} in the case $k=1$. Using the fundamental
    theorem of calculus for Lebesgue integrable functions then yields
    that $t \mapsto Z_{\Lambda, \phi^{(t)}}(\lambda)$ is absolutely
    continuous on bounded intervals and almost everywhere
    differentiable with
    \begin{equation*}
        \frac{\partial Z_{\Lambda, \phi^{(t)}}(\lambda)}{\partial t} 
        = \int_{\Lambda} \lambda e^{-\phi^{(t)}(x)} Z_{\Lambda, \phi^{(x)}}(\lambda \mid x) \vol_t(\diff x) g(t) \diff t.
    \end{equation*}
    Further, observe that $Z_{\Lambda, \phi^{(\infty)}}(\lambda) = \lim_{t \to \infty} Z_{\Lambda, \phi^{(t)}}(\lambda) = 1$ and $Z_{\Lambda, \phi^{(0)}}(\lambda) = Z_{\Lambda, \phi}(\lambda)$.
    Applying \Cref{lemma:ftc} similarly to the proof of \Cref{thm:itegral_identity}, using that $\Phi$ is closed under modification and $\delta$-zero-free, yields
    \begin{align*}
        \log(Z_{\Lambda, \phi}(\lambda))
        &= - \int_0^{\infty} \frac{1}{Z_{\Lambda,\phi^{(t)}}(\lambda)}\frac{\partial Z_{\Lambda, \phi^{(t)}}(\lambda)}{\partial t} \diff t \\
        &= - \int_{\R_{\ge 0}} \int_{\Lambda} \lambda e^{-\phi^{(t)}(x)} \frac{Z_{\Lambda, \phi^{(x)}}(\lambda \mid x)}{Z_{\Lambda, \phi^{(t)}}(\lambda)} \vol_t(\diff x) g(t) \diff t.
    \end{align*}
    Using that $Z_{\Lambda, \phi^{(t)}}(\lambda) = Z_{\Lambda, \phi^{(x)}}(\lambda)$ for almost all $t$ and $\vol_t$-almost all $x$ gives  
    \begin{align*}
        \log(Z_{\Lambda, \phi}(\lambda))
        &= - \int_{\R_{\ge 0}} \int_{\Lambda} \lambda e^{-\phi^{(t)}(x)} \frac{Z_{\Lambda, \phi^{(x)}}(\lambda \mid x)}{Z_{\Lambda, \phi^{(x)}}(\lambda)} \vol_t(\diff x) g(t) \diff t \\
        &= \int_{\Lambda} e^{-\phi^{(x)}(x)} \kappa^{(x)}(x) \vol(\diff x).\qedhere
    \end{align*}
\end{proof}

We denote by $\Phi_R$ the collection of all repulsive potentials with range at most $R$. 
Our main result is the following.
\begin{theorem} \label{thm:bounded_range} Suppose $(\X, d, \vol)$
  satisfies \Cref{as:continuity}. Let
  $B_r \coloneqq \sup_{x \in \X} \vol(\ball_x(r))$.  For $R\geq 0$ and
  $\lambda \in \C$ with $\absolute{\lambda} < \frac{1}{e B_R}$, there
  is a $C \ge 0$ such that for all bounded, measurable
  $\Lambda \subseteq \X$ and $\phi \in \Phi_R$ it holds that
  $\absolute{\log(Z_{\Lambda, \phi}(\lambda))} \le C \cdot
  \vol(\Lambda)$.
\end{theorem}

The last ingredient needed for the proof of the theorem beyond 
\Cref{thm:itegral_identity} and
\Cref{lemma:partition_function_identity} is the following contraction
argument.
\begin{lemma}
  \label{claim:boundedRange}
  Let $\Lambda \subseteq \X$ be bounded and measurable and
  $\lam \in \mathbb C$ with $|\lam| \le \frac{1 - \varepsilon}{e B_R}$
  for some $\varepsilon > 0$.  For all $\phi \in \Phi_{R}$ and all
  $f \colon \bigcup_{k \in \N} \X^k \to \mathbb C$ with
  $|f(\V{x})| \le B_R^{-k}$ for $\vol^k$-almost all
  $\V{x}\in \Lam^{k}$ it holds that
    \begin{equation*}
        \left | \lam \exp \left(- \sum_{k \in \N} \frac{1}{k!} \int_{\Lambda^k} \left(1 - e^{-\phi(y, \V{x})}\right) \cdot e^{-H(\V{x} \mid y_{\prec \V{x}})} f(\V{x}) \vol^k(\diff \V{x})\right)   \right |  \le   \frac{1 - \varepsilon}{B_R} \,.
    \end{equation*}
\end{lemma}

\begin{proof}
    We show that, for all $z \in [0, B_R^{-1}]$ and $\phi
    \in \Phi_{R}$, it holds that 
    \begin{equation*}
        \sum_{k \in \N} \frac{z^k}{k!} \int_{\Lambda^k} \left(1 -
          e^{-\phi(y, \V{x})}\right) \cdot e^{-H(\V{x} \mid y_{\prec
            \V{x}})} \vol^k(\diff \V{x}) \le 1, 
    \end{equation*}
    from which the claim follows. To this end, define a sequence of functions 
    \begin{equation*}
        G_{N}^{\phi}(z) \coloneqq \sum_{k = 1}^{N} \frac{z^k}{k!} \int_{\Lambda^k} \left(1 - e^{-\phi(y, \V{x})}\right) \cdot e^{-H(\V{x} \mid y_{\prec \V{x}})} \vol^k(\diff \V{x})
    \end{equation*}
    for $N \in \N$. By taking $N\to\infty$, it suffices to inductively
    show that $G_N^{\phi} \le 1$.  The base case $N = 1$ follows as
    (i) since $\phi$ is repulsive, the exponential term can be omitted
    for an upper bound and (ii) the bounded range assumption gives an
    upper bound of $B_{R}$ for the remaining integral.
    
    Now, suppose $G_N^{\phi} \le 1$ for all $\phi \in \Phi_R$ and some
    $N \ge 1$. The proof will consist of finding a potential
    $\hat\phi\in \Phi_{R}$ such that $G_{N+1}^{\phi}(z)\leq
    G_{N}^{\hat\phi}(z)$; the claim then follows by induction. To
    achieve this we will combine the terms $k=N$ and $k=N+1$ in
    $G_{N+1}$. The key identity for establishing this is that
    \begin{equation*}
      H(\V{x}, w \mid y_{\prec \V{x}, w}) = H(\V{x} \mid y_{\prec \V{x}}) + \phi(\V{x}, y) + \sum_{T \subseteq [k]} \phi(\V{x}_T, w \mid y_{\prec \V{x}, w}),
    \end{equation*}
    which follows from splitting all sub-tuples of $(\V{x}, w)$ based on whether they contain $w$ or not, and the fact that $\V{x}_{S}\prec \V{x} \prec (\V{x}, w)$ for
    $\vol^{k}$-almost all $\V{x}\in \X^{k}$ and $\vol$-almost all $w \in \X$ provided $S\subsetneq
    [k]$. Using this, we
    will show that
    \begin{equation*}
      G_{N + 1}^{\phi}(z) 
        = G_{N-1}^{\phi}(z) + \frac{z^N}{N!}  \int_{\Lambda^N} \left(1
          - e^{-\phi(y, \V{x})} \cdot (1 - \psi(\V{x}, y))\right)
        \cdot e^{-H(\V{x} \mid y_{\prec \V{x}})}  \vol^N(\diff \V{x}). 
      \end{equation*}
      where
    \begin{equation*}
      \psi(\V{x}, y) = \frac{z}{N + 1} \int_{\Lambda} \left(1 - e^{-\phi(y, \V{x}, w)}\right) \prod_{T \subseteq [N]} e^{- \phi(\V{x}_T, w \mid y_{\prec \V{x}, w})} \vol(\diff w) .
    \end{equation*}
    Indeed, this holds as
    \begin{align*}
      &\frac{z^N}{N!} \int_{\Lambda^N} \left(1 - e^{-\phi(y,
        \V{x})}\right) \cdot e^{-H(\V{x} \mid y_{\prec \V{x}})} \vol^N(\diff \V{x}) + \frac{z^{N+1}}{(N+1)!} \int_{\Lambda^{N+1}}\left(1 -
        e^{-\phi(y, \V{w})}\right) \cdot e^{-H(\V{w} \mid y_{\prec
        \V{w}})} \vol^{N+1}(\diff \V{w}) \\
      &=
      \frac{z^N}{N!}  \int_{\Lambda^N} \left[
        \left(1 - e^{-\phi(y, \V{x})}\right) \cdot e^{-H(\V{x} \mid
        y_{\prec \V{x}})} + \frac{z}{N+1} \int_{\Lambda}\left(1 -
        e^{-\phi(y, \V{x}, w)}\right) \cdot e^{-H(\V{x}, w \mid
        y_{\prec \V{x}, w})} \vol(\diff w) \right] \vol^N(\diff \V{x}) \\
      &=
        \frac{z^N}{N!}  \int_{\Lambda^N} \left(1 - e^{-\phi(y, \V{x})} \cdot (1 - \psi(\V{x}, y))\right) \cdot e^{-H(\V{x} \mid y_{\prec \V{x}})} \vol^N(\diff \V{x}) ,
    \end{align*}
    
    Since $\phi$ is repulsive with range $\le R$ and $z \in [0, B_R^{-1}]$, it holds that $0 \le \psi(x, y) \le
    1$. 
    Since $\phi(y, \V{x}, w) = 0$ if $\diam(y, \V{x}) > R$, in
    fact $\psi(x, y) \le \ind{\diam(y, \V{x}) \le R}$.
    Thus, setting
    \begin{equation*}
        \hat{\phi}(\V{w}) \coloneqq \begin{cases}
            \phi(\V{w}) \text{ if } |\V{w}| \le N \\
            \ind{\diam(\V{w}) \le R} \cdot \infty \text{ otherwise}
        \end{cases}
    \end{equation*}
    we have that $\hat{\phi}$ is repulsive, and
    \begin{equation*}
        1 - e^{-\phi(y, \V{x})} \cdot (1 - \psi(\V{x}, y)) \le 1 - e^{-\hat{\phi}(y, \V{x})}
    \end{equation*}
    for all $\V{x} \in \Lambda^{N}$, and
    $G_{N-1}^{\phi}=G_{N-1}^{\hat\phi}$.  Hence, we have
    \begin{equation*}
        G_{N + 1}^{\phi}(z) 
        \le G_{N-1}^{\hat{\phi}}(z) + \frac{z^N}{N!}  \int_{\Lambda^N} \left(1 - e^{-\hat{\phi}(y, \V{x})}\right) \cdot e^{-H_{\hat{\phi}}(\V{x} \mid y_{\prec \V{x}})}  \vol^N(\diff \V{x})
        = G_{N}^{\hat{\phi}}(z) 
    \end{equation*}
    as desired.
\end{proof}

\begin{proof}[Proof of \Cref{thm:bounded_range}]
  Fix some bounded, measurable $\Lambda \subseteq \X$ and let $A \subseteq \C$ be the set of all activities $\lambda$ such that,  at activity $\lambda$, $\Phi_R$ is
  $e^{-\vol(\Lambda)/B_R}$-zero-free on $\Lambda$ and
  $\absolute{\kappa_{\phi}(x)} \le \frac{1}{B_R}$ for all
  $\phi \in \Phi_R$ and $x \in \Lambda$.  
  Our goal is to show that $\D(1/e B_R) \subseteq A$, where $\D(t)$ is the open disk of radius $t$ around $0$ in $\C$.
  We note that this is equivalent to showing that $t^*\coloneqq \sup\{t \ge 0 \mid \overline{\D(t)} \subseteq A\} \geq (e B_R)^{-1}$.

  Towards a contraction, suppose that $t^* = \frac{1 - \varepsilon}{e B_R}$ for some $\varepsilon > 0$.
  We first observe that $t^* \ge 0$, since at activity $\lambda=0$ the modified point densities are $0$ and the partition functions are $1$.
  Together with the uniform equicontinuity of modified 1-point densities and partition functions given in \Cref{lemma:continuity}, this implies $t^{*} > 0$.  
  Further, by a similar continuity argument, we know that
  $\overline{\D(t^*)} \subseteq A$.\footnote{To see this, note that it
    follows from \Cref{lemma:continuity} that $A$ is closed.}
  However, using the telescoping product in \eqref{eq:telescoping} and \Cref{rem:closed}, we note that, for every $\phi \in \Phi_R$, $x \in \Lambda$ and $|\lambda| = t^*$, the modified $k$-point densities that appear in the recursion for $\kappa_{\phi, \lambda}(x)$ given by \Cref{thm:itegral_identity} are bounded by $B_R^{-k}$.
  Thus, applying \Cref{claim:boundedRange} yields that, for every $\absolute{\lambda} = t^{*}$, all $\phi \in \Phi_R$ and all
  $x \in \Lambda$, it holds that
  $\absolute{\kappa_{\phi, \lambda}(x)} \le \frac{1 - \varepsilon}{B_R}$.
  Again, by uniform equicontinuity, we can conclude that there is some neighborhood $\mathcal{N}(\lambda)$, uniformly in $x \in \Lambda$ and $\phi \in \Phi_R$, such that $\absolute{\kappa_{\phi}(x)} \le \frac{1}{B_R}$ and $\Phi_R$ is
  $\delta$-zero-free for some $\delta > 0$ at all activities in $\mathcal{N}(\lambda)$.  
  An application of \Cref{lemma:partitionfunctionFromDensity} then proves that $\Phi_R$ is in fact $e^{-\vol(\Lambda)/B_R}$-zero-free on $\Lambda$ at all
  activities in $\mathcal{N}(\lambda)$.
  Since these neighborhoods $\mathcal{N}(\lambda)$ for activities $\absolute{\lambda} = t^*$ give an open cover of the compact set $\overline{\D(t^*)} \setminus \D(t^*)$, we can pass to a finite sub-cover to find some $t' > t^{*}$ such
  that $\overline{\D(t')} \subseteq A$, which contradicts the
  definition of $t^{*}$.
\end{proof}

\section*{Acknowledgments}

TH thanks the Hausdorff
Institute for Mathematics for its hospitality, where part of this
work was carried out (funded by the Deutsche Forschungsgemeinschaft under Germany's Excellence Strategy – EXC-2047/1 – 390685813).  MP was funded by the HPI Research School on Foundations of AI. WP supported in part by NSF grant DMS-2348743.

\appendix

\section{Proof of Lemma~\ref{lemma:ordering}} \label{apx:ordering}

\begin{proof}[Proof of \Cref{lemma:ordering}]
    For part \ref{lemma:ordering:measurability} symmetry is
    immediate. Measurability is shown by induction, using
    $D_{k+1}(x_1, \dots, x_{k+1}) = D_k(x_1, \dots, x_k) + d(z,
    x_{k+1})$. 

    For \ref{lemma:ordering:continuity} we need to show that for all
    Lebesgue null sets $U \subset \R$, it holds that
    $\vol^{k} \circ D_k^{-1} (U) = 0$ for all $k\geq 1$.  For $k=1$,
    this follows from \Cref{as:continuity}.  Inductively, if the
    statement holds for some $k \in \N$, then
    \begin{align*}
        \vol^{k+1} \circ D_{k+1}^{-1}(U)
        &= \int_{\X^{k+1}} \ind{D_{k+1}(\V{x}) \in U} \vol^{k+1}(\diff \V{x}) \\
        &= \int_{\X} \int_{\X^{k}} \ind{D_{k}(\V{x}) \in U - d(z, y)} \vol^{k}(\diff \V{x}) \vol(\diff y) ,
    \end{align*}
    where $U - d(z, y)$ is the translation of the set $U$ by $- d(z, y)$.
    Since this translate is again a Lebesgue null set, it follows from
    the induction hypothesis that the inner integral evaluates to
    $\vol^{k} \circ D_k^{-1} (U - d(z, y)) = 0$, concluding the
    induction.

    For \ref{lemma:ordering:maximum}, we note that the induced ordering is trivially irreflexive, asymmetric, and transitive.
    It remains to show that, for $\vol^k$-almost all $\V{x} \in \X^k$ and all $S \neq T$ it holds that $D(\V{x}_S) \neq D(\V{x}_T)$.
    We define 
    \[
        E_k \coloneqq \{\V{x} = (x_1, \dots, x_k) \in \X^k \mid \text{ $\exists S,T \subseteq [k], S \neq T$ such that $D(\V{x}_S) = D(\V{x}_T)$}\},
    \]
    and prove that $\vol^{k}(E_k) = 0$ by induction on $k$.
    For $k = 1$, this follows from \Cref{as:continuity}.
    Suppose the statement holds for some $k \in \N$.
    For $\V{x} = (x_1, \dots, x_k) \in \X^k$, define
    \[
        F_{\V{x}} \coloneqq \big\{y \in \X \mid d(z, y) \in \bigcup_{S, T \subseteq [k]: S \neq T} \{\absolute{D(\V{x}_S) - D(\V{x}_T)}\}\big\}.
    \]
    A necessary conditions for $(x_1, \dots, x_{k+1}) \in E_{k+1}$ is that $(x_1, \dots, x_k) \in E_k$ or $x_{k+1} \in F_{x_1, \dots, x_k}$.
    Thus (applying a union bound) we have
    \[
        \vol^{k+1}(E_{k+1}) \le \int_{\X} \int_{\X^k} \ind{\V{x} \in E_k} \vol^{k}(\diff \V{x}) \vol(\diff x_{k+1}) + \int_{\X^k} \int_{\X} \ind{x_{k+1} \in F_{\V{x}}} \vol(\diff x_{k+1}) \vol^k(\diff \V{x}) . 
    \]
    By the induction hypothesis, the first term is $0$ and, by
    \Cref{as:continuity}, the second term is $0$. 
\end{proof}

\section{Proof of Lemma~\ref{lemma:ftc}} 
\label{sec:lemftc}
At first glance, \Cref{lemma:ftc} appears to follow from taking the
logarithm of both sides of the identity and then applying the
fundamental theorem of calculus.  However, under the conditions of the
statement, $\log f (z)$ might be discontinuous for any branch of the
complex logarithm.  Hence, the proof of \Cref{lemma:ftc} is slightly
more technical and requires some auxiliary claims.

To this end, we start with a recap of the basic definitions related to line integrals.
For $U \subseteq \C$ connected and open, a map $\gamma\colon [a, b] \to U$ is called a \emph{rectifiable path} in $U$ if it is continuous and has bounded variation.
Given a continuous function $f\colon U \to \C$, the line integral $\int_{\gamma} f(z) \diff z$ is defined as the limit of sums of the form
\begin{equation*}
    \sum_{j = 1}^n f(\gamma(t_j^*)) \cdot (\gamma(t_j) - \gamma(t_{j-1})) ,
\end{equation*}
where $a = t_0 < \dots < t_n = b$, $t_j^* \in [t_{j-1}, t_j]$ and the
limit is taken as the mesh size
$\max_{j \in [n]} \absolute{t_j - t_{j-1}}$ goes to $0$. See
\cite[Chapter IV, Theorem 1.4]{conway1978} for a proof of convergence.
We note that every absolutely continuous function
$\gamma\colon [a, b] \to U$ is a rectifiable path (see \cite[Lemma
3.34]{folland1999real}).  Our first ingredient is the following
identity for line integrals over absolutely continuous paths.
\begin{lemma}
  \label{claim:path_intergal}
  Suppose $\gamma\colon [a, b] \to U$ is an absolutely continuous path
  in a connected open set $U \subseteq \C$ and $f\colon U \to \C$ is
  continuous then
  \begin{equation*}
    \int_{\gamma} f(z) \diff z = \int_a^b f(\gamma(s)) \gamma'(s) \diff s,
  \end{equation*}
  where the integral on the right-hand side should be understood as a
  Lebesgue integral and $\gamma'$ is almost everywhere a derivative of
  $\gamma$.
\end{lemma}
While we expect \Cref{claim:path_intergal} is well-known, it is
typically presented for piecewise smooth $\gamma$, in which case the
right-hand side can be treated as a Riemann integral. For completeness
we give a proof of the more general claim. 
\begin{proof}[Proof of \Cref{claim:path_intergal}]
  As $\gamma$ is absolutely continuous, it is almost everywhere
  differentiable and the derivative $\gamma'$ is Lebesgue integrable
  on $[a, b]$~\cite[Theorem 7.20]{rudin1987real}.  We must show 
  \begin{equation*}
    \sum_{j = 1}^n f(\gamma(t_j^*)) \cdot (\gamma(t_j) - \gamma(t_{j-1})) ,
  \end{equation*}
  converges to $\int_a^b f(\gamma(s)) \gamma'(s) \diff s$ as the mesh
  size $\max_{j \in [n]} \absolute{t_j - t_{j-1}}$ approaches $0$. 
  Note that for all $a \le t_{j-1} < t_{j} \le b$ the
  second fundamental theorem of calculus for Lebesgue integrals (see
  \cite[Theorem 7.20]{rudin1987real}) yields $ \gamma(t_j) -
  \gamma(t_{j-1}) = \int_{t_{j-1}}^{t_{j}} \gamma'(s) \diff s$, and
  hence 
    \begin{equation*}
          \absolute{\int_a^b f(\gamma(s)) \gamma'(s) \diff s -  \sum_{j = 1}^n f(\gamma(t_j^*)) \cdot (\gamma(t_j) - \gamma(t_{j-1}))}
              \le \sum_{j = 1}^n \int_{t_{j-1}}^{t_j} \absolute{f(\gamma(s)) - f(\gamma(t_j^*))} \cdot \absolute{\gamma'(s)} \diff s.
    \end{equation*}
    Note that $f \circ \gamma$ is uniformly continuous on $[a, b]$ by compactness.
    Hence for $\varepsilon > 0$ we may choose $\max_{j \in [n]} \absolute{t_j - t_{j-1}}$ small enough such that
    \begin{equation*}
        \sum_{j = 1}^n \int_{t_{j-1}}^{t_j} \absolute{f(\gamma(s)) - f(\gamma(t_j^*))} \cdot \absolute{\gamma'(s)} \diff s \le \varepsilon \cdot \int_a^b \absolute{\gamma'(s)} \diff s
    \end{equation*}
    which implies the result since $\gamma'$ is integrable. 
\end{proof}

\begin{lemma}
  \label{claim:logFromLineintegral}
    Let $\gamma\colon [a, b] \to U$ be a rectifiable path in a
    connected open set $U \subseteq \C \setminus \{0\}$. Then
    \begin{equation*}
        \exp \left( \int_{\gamma} \frac{1}{z} \diff z \right) = \frac{\gamma(b)}{\gamma(a)} . \qedhere
    \end{equation*}
\end{lemma}

\Cref{claim:logFromLineintegral} says that $\int_{\gamma} \frac{1}{z} \diff z$ is a logarithm of $\gamma(b)/\gamma(a)$.
This is a standard result if $U \subseteq \C \setminus \{0\}$ is simply connected and $\gamma$ is piecewise smooth~\cite[Chapter III, §6]{lang2013complex}.
By approximating rectifiable paths by polygonal paths, we lift this to connected open sets.
\begin{proof}[Proof of \Cref{claim:logFromLineintegral}]
  By \cite[Chapter IV, Lemma 1.19]{conway1978}, it holds that, for
  every $\varepsilon > 0$, we can find some polygonal path
  $\Gamma\colon [a, b] \to U$ with $\Gamma(a) = \gamma(a)$,
  $\Gamma(b) = \gamma(b)$ and
  $\absolute{\int_{\gamma} \frac{1}{z} \diff z - \int_{\Gamma}
    \frac{1}{z} \diff z} < \varepsilon$.  Hence it suffices to show
  $\exp(\int_{\Gamma} \frac{1}{z} \diff z) = \Gamma(b)/\Gamma(a)$.
  Let $\Gamma_{i}$ denote the set of line segments comprising
  $\Gamma$. Since each $\Gamma_{i}$ avoids $0$, we can find some
  simply connected open set $U_i \subseteq \C \setminus \{0\}$ that
  contains $\Gamma_i$.  Since each $\Gamma_i$ is smooth,
  $\exp(\int_{\Gamma_i} \frac{1}{z} \diff z) =
  \Gamma(t_i)/\Gamma(t_{i-1})$.  Since
  $\int_{\Gamma} \frac{1}{z} \diff z = \sum_{i=1}^{k}\int_{\Gamma_i}
  \frac{1}{z} \diff z$ taking the exponential of both sides and
  collapsing the resulting telescopic product concludes the proof.
  \end{proof}

\begin{proof}[Proof of \Cref{lemma:ftc}]
  Set $\gamma\colon [0, t] \to \C, s \mapsto \frac{f(s)}{f(0)}$. Note
  that $\gamma(0) = 1$ and $\gamma(t) = \frac{f(t)}{f(0)}$.  By
  absolute continuity of $s \mapsto f(s)$ and the fact that $f$ avoids
  $0$, $\gamma$ is a rectifiable path in $\C \setminus \{0\}$.  Hence
  by \Cref{claim:logFromLineintegral} it holds that
  $\exp ( \int_{\gamma} \frac{1}{z} \diff z ) = \frac{f(t)}{f(0)}$,
  and it suffices to prove that
  $\int_{\gamma} \frac{1}{z} \diff z = \int_{0}^{t} \frac{f'(s)}{f(s)}
  \diff s$.  Since $z \mapsto 1/z$ is continuous on $\C \setminus 0$
  and $\gamma$ is an absolutely continuous path in
  $\C \setminus \{0\}$, this follows from \Cref{claim:path_intergal}
  and the fact that $\gamma' = f'$ almost everywhere.
\end{proof}

  \section{Lower Bound for Generic Repulsive Bounded-Range Potentials}
\label{sec:sharp}

Call a potential $\phi = (\phi_m)_{m \in \N}$ a \emph{pure $k$-body
  potential} if $\phi_{m}=0$ for all $m \neq k$. This appendix proves
the following proposition.

\begin{proposition}
  \label{prop:upperBound}
  Fix $R>0$, $k\ge 2$. There exists a sequence of complete, separable
  metric spaces $(\X_n, d_n,\vol_{n})_{n \in \N}$ satisfying
  \Cref{as:continuity} and a pure $k$-body potential $\phi^{(n)}$ of
  range $R$ such that
  \begin{enumerate}[noitemsep]
  \item $B^{(n)}_R \coloneqq \sup_{x \in \X_n} \vol_n(\{y \in \X_n
    \mid d_n(x, y) \le R\}) $ satisfies $\lim_{n \to \infty} B^{(n)}_R = \infty$. 
  \item There are $\Lambda_{n}\subseteq \X_{n}$ and 
    $\lambda_n = O\left(\frac{\log B^{(n)}_R}{B^{(n)}_R}\right)$
    such that, for every $n \in \N$,
    $Z_{\Lambda_n, \phi^{(n)}}(-\lambda_n) = 0$.
  \end{enumerate}
\end{proposition}

The takeaway of \Cref{prop:upperBound} is that $B_{R}$ is essentially
the optimal scaling for the activity if one wants to ensure the
analyticity of the pressure in a disk of positive radius. The proof of
\Cref{prop:upperBound} is by construction of an appropriate sequence
of spaces that allow us to make use of a known result about the
location of negative roots in the hypergraph hard-core model. We will
shortly define what this last object is, but first, we state two lemmas
that explains why it is useful. 
\begin{lemma}
  \label{lemma:simulateHypergraph}
  For every $\Delta \in \N$, every $k$-uniform hypergraph $G = (V, E)$
  of maximum degree $\Delta$ and every $R > 0$, there is a complete,
  separable metric space $(\X, d)$ with compatible locally-finite
  volume measure $\vol$, a pure $k$-body potential $\phi$ of range $R$
  and bounded, measurable region $\Lambda \subseteq \X$ such that
  $(\X, d, \vol)$ satisfies \Cref{as:continuity}, $B_R = \Delta + 1$
  and $Z_{\Lambda, \phi}(\lambda) = Z_{G}(e^{\lambda} - 1)$ for all
  $\lambda \in \C$, where $Z_{G}$ denotes the independence polynomial
  of $G$.
\end{lemma}

Given \Cref{lemma:simulateHypergraph}, which is proven below,
\Cref{prop:upperBound} follows from the following theorem concerning
the zeros of the hypergraph independence
polynomial. 
\begin{theorem}[{\cite[Theorem~1.2]{zhang2025hypergraph}}]
  \label{lemma:hypergraphZeros}
  For every $k \ge 2$, there is a sequence of $k$-uniform hypergraphs
  $(G_n)_{n \in \N}$ with maximum degrees $\Delta_n \to \infty$, such
  that for some sequence of vertex activities
  $z_n \in O\left(\frac{\log \Delta_n}{\Delta_n}\right)$ such that
  $Z_{G_n}(-z_n) = 0$.
\end{theorem}

\begin{proof}[Proof of \Cref{prop:upperBound}]
 Combine \Cref{lemma:simulateHypergraph} and \Cref{lemma:hypergraphZeros}.
\end{proof}

The following definitions concerning hypergraphs and the hypergraph
independence polynomial are all that will be needed to prove
\Cref{lemma:simulateHypergraph}. A hypergraph $G=(V,E)$ generalizes a
graph by allowing for $E$ to be a set of subsets $e\subset V$ of
arbitrary cardinality $|e|\geq 2$. It is $k$-uniform if $|e|=k$ for
all $e\in E$. The degree of a vertex $v$ is the number of hyperedges
containing $v$. A subset $I\subset V$ is an independent set if
$|I\cap e|<|e|$ for all $e\in E$. Lastly,
$Z_{G}(z) = \sum_{I\subset V} z^{|I|}\ind{I\text{ independent} }$.
\begin{proof}[Proof of \Cref{lemma:simulateHypergraph}] 
  It will be convenient to assume $V = [N]$ for some $N \in \N$. Let
  $d_G$ denote the hop distance in $G$. That is, $d_{G}(v,v')=k$ if
  the shortest sequence $\{e_{i}\}$ of hyperedges such that $v\in e_{1}$, $v'\in
  e_{k}$ and $e_{i}\cap e_{i+1}\neq\emptyset$ has length $k$. Set
  $\X \coloneqq \bigcup_{j \in [N]} [2j, 2j+1]$,
  \begin{equation}
    \label{e:UBd}
    d(x, y) \coloneqq
    \begin{cases}
      \frac{R}{10} \cdot |x-y| &\text{ $x, y \in [2j, 2j+1]$},\\
      \frac{R}{10} \cdot \left( 8d_G\left(\lfloor x \rfloor/2, \lfloor y
        \rfloor/2\right) + \frac{1}{10} \cdot |x - \lfloor x \rfloor|
      + \frac{1}{10} \cdot |y - \lfloor y \rfloor|\right) &\text{ otherwise,}
    \end{cases}
  \end{equation}
  and let $\vol$ be Lebesgue measure on $\R$ restricted to $\X$. In
  the right-hand side of~\eqref{e:UBd}, and below, we have made
  use of our above identification of integers with vertices of $G$.

  We first verify the structural properties of $(\X,d,\vol)$ and \Cref{as:continuity}. It
  is easily checked that $d$ is indeed a metric on $\X$.  To see that
  $(\X, d)$ is complete, we note that every sequence that is Cauchy
  with respect to $d$ must eventually be contained in some interval
  $[2j, 2j+1]$ for some $j \in [N]$. Convergence in $[2j,2j+1]$ is then inherited
  from the completeness of $(\R,|\cdot|)$ since $d$ is a rescaling of
  the Euclidean distance. This verifies completeness. 

  Next, note that the Borel algebra on $(\X, d)$ is inherited from the
  Borel algebra $\mathcal{B}(\R)$ on $\R$, i.e., consists of sets $\{ B\cap \X \mid B\in
  \mathcal{B}(\R)\}$. 
  To see this, note that every Borel subset of $B \subseteq \X$ in
  $\R$ can be partitioned into sets $B_j = B \cap [2j, 2j+1]$ for
  $j \in [N]$.  It suffices then to show that $B_j$ is also in the
  Borel algebra generated from $d$, which again follows from the
  equivalence of metrics on such intervals.

  Since $\vol$ is Lebesgue measure on $\R$ restricted to $\X$, we can
  observe that $B_R = \Delta + 1$.  To see that $(\X, d, \vol)$
  satisfies \Cref{as:continuity}, it is easy to check that, for every
  $x \in \X$, the map
  $r \mapsto \vol(\{y \in \X \mid d(x, y) \le r\})$ is an absolutely
  continuous function on $\R$.

  We now define $\phi$. Set, for $(x_1, \dots, x_k) \in \X^k$,
  $\phi_k(x_1, \dots, x_k) = \infty \cdot \ind{E}(\{\lfloor x_1
  \rfloor/2, \dots, \lfloor x_k \rfloor/2\})$. Otherwise set
  $\phi_j = 0$ for all $j \neq k$.  By definition, $\phi$ is a pure
  $k$-body potential and has range at most $R$.
    
  It remains to show that
  $Z_{\Lambda, \phi}(\lambda) = Z_{G}(e^{\lambda} - 1)$ for all
  $\lambda \in \C$. Given $(x_{1},\dots, x_{m})$, let
  $I(x_1, \dots, x_m) \coloneqq \{\lfloor x_i \rfloor/2 \mid i \in
  [m]\}$ denote the corresponding vertex set in $G$. Note that
  $H(x_{1},\dots, x_{m})=0$ if $I(x_{1},\dots, x_{m})$ is an
  independent set in $G$, and $H(x_{1},\dots, x_{m})=\infty$
  otherwise.  Evidently $|I(x_1, \dots, x_m)| \le m$, with
  $|I(x_1, \dots, x_m)| = m$ if and only if each all $x_i$ are in
  distinct intervals $[2j, 2j+1]$ for $j \in [N]$.  Recall that for
  all finite sets $A, B$ there are $|B|! \stirling{|A|}{|B|}$
  subjective maps from $A$ to $B$, where $\stirling{a}{b}$ denotes
  Stirling number of the second kind. Using this, observe that for
  every non-empty independent set $S$ in $G$ and every $m \ge |S|$, it
  holds that
  \begin{equation*}
    \int_{\Lambda^m} \ind{I(\V{x}) = S} \vol(\diff \V{x}) = |S|!
    \stirling{m}{|S|}.
  \end{equation*}
  Combining the observations above, we have
  \begin{equation*}
    Z_{\Lambda, \phi}(\lambda) 
    = 1 + \sum_{m \in \N} \frac{\lambda^{m}}{m!} \int_{\Lambda^m} e^{-H(\V{x})} \vol(\diff \V{x}) = 1 + \sum_{m \in \N} \frac{\lambda^{m}}{m!}  \sum_{\ell = 1}^{m} \ell! \stirling{m}{\ell} |\mathcal{I}_{\ell}|,   
  \end{equation*}
  where $\mathcal{I}_{\ell}$ is the set of independent sets of $G$ with cardinality
  $\ell$.  Exchanging the sums and using that
  $\sum_{m=\ell}^{\infty} \frac{x^{m}}{m!} \stirling{m}{\ell} =
  \frac{(e^x - 1)^{\ell}}{\ell!}$, we obtain
  \begin{equation*}
    Z_{\Lambda, \phi}(\lambda) 
    = 1 + \sum_{\ell \in \N} |\mathcal{I}_{\ell}| \sum_{m
    =\ell}^{\infty} \frac{\lambda^{m}}{m!}  \ell! \stirling{m}{\ell}
    = 1 + \sum_{\ell \in \N} |\mathcal{I}_{\ell}| (e^{\lambda} -
    1)^{\ell}  = Z_{G}(e^{\lambda} - 1).\qedhere
  \end{equation*}
\end{proof}

\end{document}